\documentclass[12pt]{article}
\usepackage{amsthm,amsmath,indentfirst,amssymb}
\usepackage{algorithm,algpseudocode,color}
\usepackage{mathrsfs}
\usepackage[backref]{hyperref}
\usepackage[capitalise]{cleveref}
\usepackage{url}
\usepackage{pstricks}
\usepackage{tikz}
\usepackage{caption}
\usepackage{subcaption}
\usepackage{abstract}
\usepackage{arcs}
\usepackage{yhmath}
\usetikzlibrary{shapes.geometric}

\usetikzlibrary{calc, shadows.blur}
\newtheorem{lemma}{Lemma}[section]
\newtheorem{theorem}[lemma]{Theorem}

\newtheorem{observation}[lemma]{Observation}
\newtheorem{definition}[lemma]{Definition}

\begin{document}

\title{Approximation Algorithms for Minimum Sum of Moving-Distance and Opening-Costs Target Coverage Problem}
\author{\footnotesize Lei Zhao \quad Zhao Zhang \thanks{Corresponding author: Zhao Zhang hxhzz@sina.com}\\
	{\it\small School of Mathematical Sciences, Zhejiang Normal University}\\
	{\it\small Jinhua, Zhejiang, 321004, China}}
\date{}
\maketitle
\setlength{\absleftindent}{0pt}
\setlength{\absrightindent}{0pt}

\noindent{\bf Abstract.}
{\rm In this paper, we study the Minimum Sum of Moving-Distance and Opening-Costs Target Coverage problem (MinMD$+$OCTC). Given a set of targets and a set of base stations on the plane, an opening cost function for every base station, the opened base stations can emit mobile sensors with a radius of $r$ from base station to cover the targets. The goal of MinMD$+$OCTC is to cover all the targets and minimize the sum of the opening cost and the moving distance of mobile sensors. We give the optimal solution in polynomial time for the MinMD$+$OCTC problem with targets on a straight line, and present a 8.928 approximation algorithm for a special case of the MinMD$+$OCTC problem with the targets on the plane.}


\section{Introduction}
Coverage, as a fundamental issue in Wireless Sensor Networks (WSNs). This is crucial as it directly affects the quality of service the network can provide in terms of sensing accuracy, detection reliability, and overall operational effectiveness.

In the early stages of research, the use of static sensors for Target Coverage (TCOV) problem is widely studied \cite{basappa2015unit,hochbaum1985approximation,lam2023target,li2015ptas,ling_ding_constant_2012,mustafa2010improved,zhang2016approximating}. However, with the emergence of more dynamic technologies such as Unmanned Aerial Vehicles (UAVs) and other autonomous agents, mobile sensors have gained significant attention \cite{caillouet2018optimization,ko2022uav,li2011sweep,xu2021minimizing}. In mobile sensor networks, base stations play an important role, such as data integration and processing \cite{kumar2018mobile,kumar2023approximation,liang2022approximation}, or charging \cite{liang2021changeable,liang2019approximation}.

Chen {\it et al.} \cite{chen2016ptas} first integrated the concept of base stations into TCOV problem and got PTAS. In their model, the base stations and the targets are distributed on a plane, the sensor can cover the range of radius $r$, and all the sensor initial positions are at the base station. The objective is to minimize the moving distance of the sensors while covering all targets. After that, Guo {\it et al.} \cite{guo2021minsum} extended the discussion to include heterogeneous sensors with base stations, that is, the coverage radius of sensors is different.

There are a number of classic facility location problems, the most basic of which is the uncapacitated Facility Location (UFL) problem. In UFL problem, given a bipartite graph $(\mathcal{F}, \mathcal{C})$, where $\mathcal{F}$ is the set of facilities and $\mathcal{C}$ is the set of customers. For any $i$ belongs to $\mathcal{F}$, $j$ belongs to $\mathcal{C}$, $f_i$ is the opening cost of facility $i$, and $c_{ij}$ is the connection cost of facility $i$ to customer $j$. The goal is to choose to open a subset $\phi$ of $\mathcal{F}$ and connect all customers to $\phi$, ultimately minimizing the opening cost and connection cost. There is no requirement for the number of customers a facility can connect to.

In this paper, considering the location problem of the base station, different locations of base stations incur varying opening costs. Only the opened base stations are capable of
deploying sensors with radius $r$. Our objective is to identify the optimal locations for base stations from a set of candidates, emit sensors from these locations to cover the targets, and minimize the total of opening costs and the sensors' movement distance. It can be seen that when $r$ is 0, the MinMD$+$OCTC problem is an UFL problem. This is the first paper to
study the decision-making process in base station location for target coverage.

\subsection{Related works}
Liao {\it et al.} \cite{liao2014minimizing} were the first to study TCOV problem with sensor aiming at the sensor movement. In their model, it is necessary to move the mobile sensors on the plane to cover all the targets and minimize the total moving distance, the NP-hardness of TCOV is proofed. Spacially, they use the Hungarian algorithm to obtain a polynomial time algorithm for the case that the distance between any two targets is more than twice of the sensor coverage radius. Besides, they design a heuristic algorithm for general case.
There are also some articles studying the TCOV problem of minimizing the sum of sensor movement distances under limited budget constraint. For example, Nguyen {\it et al.} \cite{Nguyen2020NP} studied the TCOV problem with a limited number of sensors. Liang {\it et al.} {liang2020maximum} studied the TCOV problem with a limited movement distance $b$ for each sensor. They both obtained the $1-1/e$ approximation. Jin {\it et al.} \cite{jin2024approximation} also considered 
that the total budget has an upper bound $B$ on the basis of \cite{liang2020maximum}, obtained $\frac{1}{2v}$-approximation, where $v=\lceil \frac{mb}{B-b}\rceil$, $m$ is the number of sensor.

For TCOV problem with base station, Chen {\it et al.} \cite{chen2016ptas} studied the {\em k-Sink Minimum Movement Target Coverage} problem (k-SMMTC), where $k$ is the number of base station. In the k-SMMTC problem,
the base station plays the role of allocating coverage tasks to the sensor. They obtained PTAS within $O(n^{\frac{1}{{\epsilon}^{2}}})$ time for k-SMMTC problem.
Wongwattanakij {\it et al.} \cite{wongwattanakij2023improvedptas} got same ratio and improved its running time to $O(n^{\frac{1}{\epsilon}})$. In particular, Guo {\it et al.} \cite{guo2021minsum} considered heterogeneous sensors. They used dynamic programming to obtain the optimal solution for the targets on the straight line. 
However, none of the articles about base stations have considered the location of base stations.

For the metric UFL problem, David B. Shmoys et al. \cite{shmoys1997approximation} proposed the first constant factor approximation algorithm, achieving the approximation ratio of 3.16 using LP-rounding. Subsequently, Jaroslaw Byrka and Karen Aardal \cite{byrka2010optimal} enhanced the LP-rounding method to obtain an improved approximation ratio of 1.5. Kamal Jain et al. \cite{jain2003greedy} achieved a 1.61 approximation ratio through the dual fitting method. Li Shi \cite{li_1488_2013} further improved the approximation ratio to 1.488 by combining the methods in \cite{byrka2010optimal} and \cite{jain2003greedy}, which currently represents the best known approximation ratio.

\subsection{Contribution}
This paper studies the MinMD$+$OCTC problem. The contributions are summarized as follows:

$(\romannumeral 1)$ We use dynamic programming to give an exact algorithm for MinMD$+$OCTC$_{line}$. Especially for situations where the sensor center does not need to move to a straight line, dynamic programming can also be used to obtain the optimal solution.

$(\romannumeral 2)$ For the MinMD$+$OCTC problem, we show its NP-hardness and give a 8.928 approximation algorithm. When the sensor's coverage range is square, it can be improved to 7.44 approximation.

\section{The Algorithm}\label{se2.2}
The problem studied in this paper is formally defined as follows.

\begin{definition}[MinMD$+$OCTC]
{\rm Suppose $\mathcal{T}=\{t_{1},\ldots,t_{n}\}$ is a set of targets and $\mathcal{B}=\{b_{1},\ldots,b_{m}\}$ is a set of sensors sorted in ascending order of their $x$-coordinates on the plane. The opening cost of base station $b_i$ is $c_i$. An opened base station can emit mobile sensors to some locations for the service of coverage. The number of mobile sensors emitted from a base station is unlimited. Each mobile sensor has a covering radius $r$ and can cover all those targets within distance at most $r$ from it. The problem is to determine a subset of base stations to be opened and the locations for the mobiles sensors emitted from opened base stations such that all targets are covered and the sum of moving-distances and opening-costs is minimized. When all targets lie on a line $L$, and the sensor center moves to a straight line for coverage, this problem is denoted as MinMD$+$OCTC$_{line}$.}
\end{definition}

\subsection{Polynomial-Time Algorithm for MinMD$+$OCTC$_{line}$}\label{sec_ploy_alg_for_min}

For convenience, we first introduce a lemma to illustrate that the MinMD$+$OCTC$_{line}$ problem can be equivalent to an MinMD$+$OCTC$_{line}$ problem that only considers the base station on one side of the straight line $L$.

\begin{lemma}\label{Lem_instance_transformation} 
	Given an instance of the  MinMD$+$OCTC$_{line}$ problem $\mathcal{I}=(\mathcal{T},\mathcal{B})$, where the coordinates of base station $b_k$ are $(x(b_k),y(b_k))$, and the opening cost is $c_k$. Construct a new instance $\mathcal{I}'=(\mathcal{T},\mathcal{B}')$, where the coordinates of base station $b_k'$ is $(x(b_k),|y(b_k)|)$, and the opening cost is $c_k$. Then $opt(\mathcal{I})=opt(\mathcal{I}').$
\end{lemma}

\begin{proof}
	There is an optimal solution $OPT(\mathcal{I})$ for instance $\mathcal{I}$. We next use $OPT(\mathcal{I})$ to construct a feasible solution $APX(\mathcal{I}')$ for instance $\mathcal{I}'$. First, base station $b_k'$ is opened if and only if base station $b_k$ is opened. Secondly, for any sensor $s$ with coordinates $(x(s),0)$ emitted by the base station $b_k$ in $OPT(\mathcal{I})$, place a sensor $s'$ emitted by $b_k'$ at the same position in the $\mathcal{I}'$ instance. Let $d_k(s)$ is the movement distance of sensor $s$ emitted by base station $b_k$, one has $d_k(s)=d_{k'}(s')$. Obviously, the solution after the construction is a feasible solution of $\mathcal{I}'$, and its value is $opt(\mathcal{I})$, which means $opt(\mathcal{I}' ) \leq opt(\mathcal{I})$. Similarly, we can get $opt(\mathcal{I} ) \leq opt(\mathcal{I}')$. Therefore, $opt(\mathcal{I}) = opt(\mathcal{I}')$.
\end{proof}
As can be seen from the proof of Lemma \ref{Lem_instance_transformation}, $OPT(\mathcal{I}')$ can be easily converted into $OPT(\mathcal{I})$. Next, we only consider the case where the base station is above the straight line $L$.
In the following, we present a key observation of algorithm.

\begin{lemma}\label{Lem1} 
There exists an optimal solution to the MinMD$+$OCTC$_{line}$ instance such that the locations of those mobile sensors emitted by a same base station are consecutively arranged on the line. That is, if there are two sensors $s,s'$ emitted by a same base $b_k$, then any sensor between $s,s'$ on the line are emitted by $b_k$.  Furthermore, if $k<k'$, the sensor emitted by $b_k$ must be on the left of the sensor emitted by $b_{k'}$.
\end{lemma}

\begin{proof}
Let $C_{s,b_k}$ be the disk centered at center of sensor $s$ with radius $d_k(s)$. Suppose sensor $s$ is emitted by base station $b_k$. In order to minimize the moving distance, $b_k$ must be the base station nearest to $s$ among all opened base stations. 
Then
\begin{equation}\label{eq0414-3}
\mbox{there is no other base station lying in the interior of disk $C_{s,b_k}$.}
\end{equation}

\begin{figure}[htbp]
	\centering
	\begin{tikzpicture}
		
		\coordinate (A) at (0,0);
		\coordinate (B) at (5,0);
		\coordinate (C) at (1.5,2);
		\coordinate (D) at (2,0);
		
		\draw (A) -- (B);
		\draw[dashed] (A) -- (C);
		\draw[dashed] (B) -- (C);

		\filldraw[black] (A) ++(-0.1,0) -- ++(0.1,0.2) -- ++(0.1,-0.2) -- cycle; 
		\filldraw[black] (B) ++(-0.1,0) -- ++(0.1,0.2) -- ++(0.1,-0.2) -- cycle; 
		\filldraw[black] (D) ++(-0.1,0) -- ++(0.1,0.2) -- ++(0.1,-0.2) -- cycle;
		
		\filldraw[black] (C) ++(-0.15,-0.15) rectangle ++(0.3,0.3);

		\node at (A) [below left] {$\text{s}$};
		\node at (B) [below right] {$\text{s}'$};
		\node at (C) [above] {$\text{b}_k$};
		\node at (C) [below] {$\text{A}$};
		\node at (D) [below right] {$\text{s}''$};

		\draw[dashed] (A) let \p1 = ($(C)-(A)$) in circle ({veclen(\x1,\y1)});

		\draw[dashed] (B) let \p1 = ($(C)-(B)$) in circle ({veclen(\x1,\y1)});
		\path let \p1 = ($(C)-(D)$) in \pgfextra{\xdef\DCRadius{veclen(\x1,\y1)}};

		\fill[gray, opacity=0.2] (D) circle (\DCRadius);
		\draw[dashed] (D) circle (\DCRadius); 
		
	\end{tikzpicture}
	\caption{An illustration for the proof of consecutive arrangement of sensors. }\label{fig0414-2}
\end{figure}
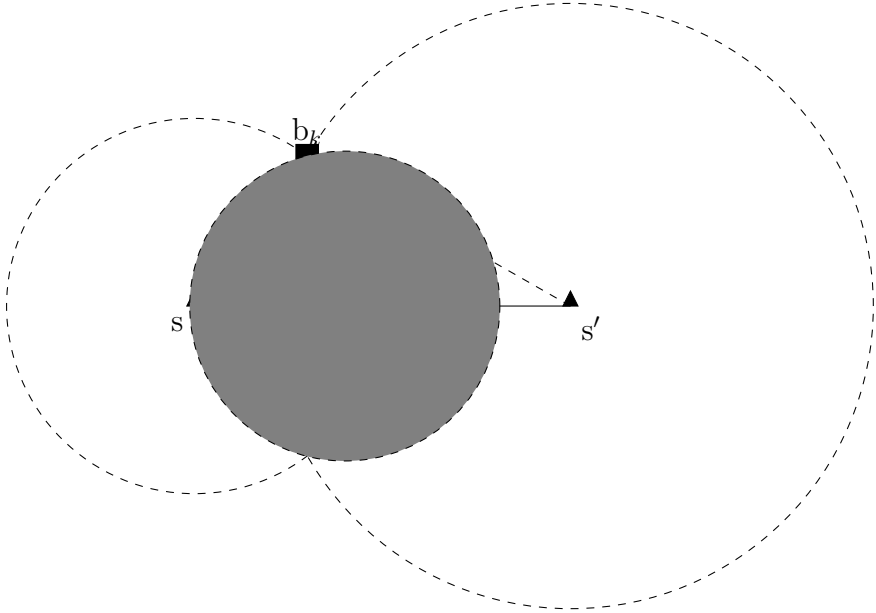

Suppose $s,s'$ are two mobile sensors emitted by a same base station $b_k$ such that there is a sensor $s''$ lying between them that is emitted by another base station $b_{k'}\neq b_k$ (see Fig. \ref{fig0414-2} for an illustration). By observation \eqref{eq0414-3},
\begin{equation}\label{eq0419-1}
\mbox{$b_{k'}$ cannot lie in the interior of $C_{s,b_k}\cup C_{s',b_k}$.}
\end{equation}
On the other hand, since $\|d_{k'}(s'')\|\leq \|d_k(s'')\|$,
\begin{equation}\label{eq0419-2}
\mbox{$b_{k'}$ must lie in disk $C_{s'',b_k}$.}
\end{equation}
Properties \eqref{eq0419-1} and \eqref{eq0419-2} can hold simultaneiously only when $\|d_{k'}(s'')\|= \|d_k(s'')\|$. This means that base station $b_{k'}$ can only be located at the position of  base station $b_k$ in the figure \ref{fig0414-2}. 
For any sensor emitted by base station $b_{k'}$, we can let base station $b_k$ emit a sensor to the same location, and the moving distance cost is the same. Therefore base station $b_{k'}$ does not need to be opened to reduce the opening cost. This shows that any sensor between $s$ and $s'$ is emitted by $b_k$.

Consider any two open base stations $b_k,b_{k'}$. Draw a midpoint perpendicular for the segment between $b_k$ and $b_{k'}$. This midpoint perpendicular can divide line $L$ into two segments $l_1$ and $l_2$, where $l_1$ (resp. $l_2$) is the subline of $L$ consisting of points on the left (resp. right) of vertical bisector line, which can be shown in figure \ref{fig_Mid_line}. By the property of midpoint perpendicular, the sensors located on $l_1$ (resp. $l_2$) are emitted by $b_k$ (resp. $b_{k'}$). Therefore, if $k<k'$, the sensor emitted by $b_k$ must be on the left of the sensor emitted by $b_{k'}$.
  
  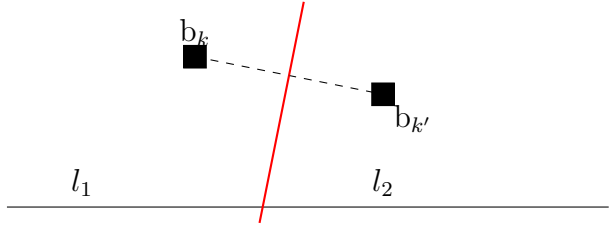
\begin{figure}[htbp]
  	\centering
  	\begin{tikzpicture}
  		
  		\coordinate (A) at (0,0);
  		\coordinate (B) at (8,0);
  		\coordinate (C) at (2.5,2);
  		\coordinate (D) at (5,1.5);
  		\coordinate (E) at (1,0);
  		\coordinate (F) at (5,0);
  		
  		\draw (A) -- (B);
  		\draw[dashed](C)--(D);

  		\filldraw[black] (D) ++(-0.15,-0.15) rectangle ++(0.3,0.3); 
  		\filldraw[black] (C) ++(-0.15,-0.15) rectangle ++(0.3,0.3); 
  		\node at (C) [above] {$\text{b}_k$};
  		\node at (D) [below right] {$\text{b}_{k'}$};
  		\node at (E) [above] {$l_1$};
  		\node at (F) [above] {$l_2$};
  		
  		\coordinate (MidCD) at ($ (C)!0.5!(D) $);
  		
  		\draw[red, thick] (MidCD) -- ($(MidCD)!1cm!90:(D)$); 
  		\draw[red, thick] (MidCD) -- ($(MidCD)!2cm!-90:(D)$);

  	\end{tikzpicture}
  	\caption{An illustration for the proof of sensors orderliness. }\label{fig_Mid_line}
  \end{figure}
  
\end{proof}

Suppose the targets are ordered as $t_1,\ldots,t_n$ from left to right on the line. By Lemma \ref{Lem1}, every open base station emits the sensors located in a consecutive interval. Therefore, we have the following observation for the targets on the line.

\begin{observation}\label{Observation_target_interval}
If $t_i$ and $t_j$ are covered by the sensor emitted by the base station $b_k$, all the targets between $t_i$ and $t_j$ are covered by the sensors emitted by $b_k$.
\end{observation}

 Let the coordinate of $t_i$ is $(x(t_i),0)$. From observation \ref{Observation_target_interval}, we can say that base station $b_k$ is responsible for the interval $I_k=[x(t_i),x(t_j)]$, and get a conclusion similar to Lemma \ref{Lem1}, that is, for any base station, the target interval it is responsible for is consecutive, and the intervals responsible for different base stations are also ordered.

In order to find out an optimal solution satisfying Lemma \ref{Lem1} and Observation \ref{Observation_target_interval}, it suffices to determine the interval $I_k$ that each opened base station $b_k$ is responsible for. This can be realized using a dynamic programming method.

{\bf Dynamic programming for MinMD$+$OCTC$_{line}$.}
By Lemma \ref{Lem1}, we can naturally define state variables as follows. For $i=1,\ldots,n$, let $c(\mathcal B_k,\mathcal T_i)$ be the minimum cost of using some base stations in $\mathcal B_k=\{b_1,...,b_k\}$ to be responsible for the targets in $\mathcal T_i=\{t_1,\ldots,t_i\}$, where the cost includes both the opening cost and the moving distance cost. For $k\in\{1,\ldots,m\}$ and $i,j\in \{1,\ldots,n\}$ with $i\geq j$, let $d_{k}(i-j+1,i)$ be the minimum moving distance cost of the sensors covering the targets in $I_k=[x(t_{i-j+1}),x(t_i)]$ that the base station $b_k$ is responsible for. We employ a ``bi-level'' dynamic programming method. In the outer-level dynamic programming, $c(\mathcal B_k,\mathcal T_i)$ are determined based on the values of $d_{k}(i-j+1,i)$. Note that $d_{k}(i-j+1,i)$ depends on the arrangement of sensors (that is, the destinations of the sensors) emitted by $b_k$. The inner-level dynamic programming is used to compute  $d_{k}(i-j+1,i)$. 

Status transition formula for $c(\mathcal B_k,\mathcal T_i)$ is

\begin{equation}\label{out_dp}
c(\mathcal B_k,\mathcal T_i)=\min\limits_{1\leq j\leq i}\{ c(\mathcal B_{k-1},\mathcal T_{i-j})+c_k+d_k(i-j+1,i),c(\mathcal B_{k-1},\mathcal T_i)\}.
\end{equation}

 with initial condition
	
	\begin{equation}\label{out_boundary}
		\left\{
		\begin{array}{lcl}
			c(\mathcal B_k,\emptyset)=0, & &  k \geq 0,\\
			
			c(\emptyset,\mathcal T_i)=\infty, & & i> 0. \\
		\end{array} \right.
	\end{equation}
	
When considering covering targets $\mathcal T_i$, due to Lemma \ref{Lem1}, if $b_i$ covers the targets in the interval $[x(t_{i-j+1}),x(t_i)]$, then for $b_k$ with $k>i$, it will not cover the targets on the left side of $t_{i-j + 1}$. Therefore, there are only two types of situations in $c(\mathcal B_k,\mathcal T_i)$. One is to add the service of base station $b_k$ to interval $[x(t_{i-j+1}),x(t_i)]$ on the basis of $c(\mathcal B_{k-1},\mathcal T_{i-j})$ for $1\leq j \leq i$, or chooses not to use $b_k$. In the latter case, the whole $T_i$ is served by $\mathcal B_{k-1}$. The optimal value will take the minimum of all these cases. 

We now consider the inner-level dynamic programming $d_{k}(i-j+1,i)$. 
We first show that the number of possible positions of the sensors associated with an interval is polynomial.
\begin{observation}\label{Ploy_sensor_position}
	For any base station $b_k$ and a interval $I_k=[x(t_{i-j+1}),x(t_i)]$ $(1 \leq j \leq i )$ 
	, the possible positions of the sensor emitting from base station $b_k$ within the target covering this interval are at most $2j+1$. 
\end{observation}
\begin{proof}
	Assume that the coordinates of base station $b_k$ are $(x(b_k), y(b_k))$. For any sensor $s$ transmitted by base station $b_k$, assume that the final coordinates are $(x(s), 0)$. If $x(s)< x(b_k)$, there must be a target on the left boundary of the sensor. Otherwise, sensor $s$ can reduce the cost of moving distance by moving to the right. The same is true for $x(s)> x(b_k)$. 
	$x(s)= x(b_k)$ corresponds to the sensor emitting vertically on the straight line. Because a target has two situations on the sensor boundary, there are at most $2j+1$ situations above.
\end{proof}

Record the $2j+1$ possible position sensors on the interval $[x(t_{i-j+1}),x(t_i)]$ that $b_k$ is responsible for covering as $G_{k}(i-j+1,i)$, and $S_i$ is the set of sensors covering target $i$. Denote the distance between any sensor $s$ to any target $b_k$ as $d_k(s)$.
 
$t_{i'_s}$ is the maximum subscript target on the left side of the sensor $s$ coverage.  We consider covering the targets on the interval $[x(t_{i-j+1}),x(t_i)]$ with the least number of sensors.

  The dynamic programming as follows, in which $0\leq i \leq n$.

\begin{equation}\label{in_dp}
	d_{k}(i-j+1,i)=\min\limits_{s\in G_{k}(i-j+1,i)\cap S_{i}}\{d_{k}(s)+d_{k}(i-j+1,i'_s)\}.
\end{equation}

The boundary condition of dynamic programming is

\begin{equation}\label{in_boundary}
	\begin{array}{lcl}
		
		d_{k}(l,i)=0, & & i<l. \\
	\end{array}
\end{equation}

Next we prove the optimality and time of dynamic programming.
\begin{theorem}\label{line_opt}
	There is an $O(n^{4}m)$ algorithm to solve MinMD$+$OCTC$_{line}$.
\end{theorem}
\begin{proof}
It can be seen from the dynamic programming equation \ref{out_dp} that if the base station $b_k$ is responsible for the interval $[x(t_{i-j+1}),x(t_i)]$,  the set of base stations $\mathcal{B}_{k-1}$ will be responsible for the interval $[x(t_1),x(t_{i-j})]$. So in the solution we found, the intervals responsible for different base stations do not overlap. We next show that the optimal solution satisfies this property.

 For an optimal solution $OPT$ that satisfies Lemma \ref{Lem1}, assume that the base stations opened are $b_{1}',...,b_{k}'$, and the longest intervals that these base stations are responsible for are $I_1,...,I_k$. Let $I_{1}'=I_1$, $I_{2}'=I_2 \setminus I_1$, ..., $I_{k}'=I_k\setminus I_{k-1}$. The base station opened and the position of the emitted sensor remain unchanged. After such operations, there is no overlap in the intervals that the base station is responsible for. $OPT$ is a feasible solution under interval partitioning $\{I_{1}',...,I_{k}'\}$. However, dynamic programming can find the minimum value under this partition, that is, it can find the optimal solution.

	In Equation \ref{out_dp}, there are $O(mn)$ state variables to be calculated, each state variable is calculated by the sub-state variables that $O(n)$ has prepared. In addition, the inner dynamic programming will cost $O(n^{2})$, there are $O(n)$ state variables, each state variable needs to be calculated $O(n)$ times. Therefore, the time of the planning dynamic algorithm is $O(n^{4}m)$.

\end{proof}

The above question requires covering all targets. If it is only required to cover at least $K$ targets, we call it partial MinMD$+$OCTC$_{line}$ problem.
The use of the sensor depends only on the division of the target intervals, and targets within an interval can only be served by the sensors of an opened base station. Thus, partial MinMD$+$OCTC$_{line}$ problem, the only difference from the MinMD$+$OCTC$_{line}$ problem is that it covers partial targets in an interval instead of all, can be easily solved using dynamic programming, due to the retention Lemma \ref{Lem1}. 

Here we give dynamic programming for the partial MinMD$+$OCTC$_{line}$ problem. The goal of partial MinMD$+$OCTC$_{line}$ problem is to cover at least $K$ points so that the total cost is minimal. Since Lemma \ref{Lem1} is satisfied, outer dynamic programming, like Equation \ref{out_dp}, aims to partition intervals. The only difference is that in Equation \ref{in_dp}, in partial MinMD$+$OCTC$_{line}$ problem, some targets are not considered to cover. Let $d_k (l,i,K)$ be the minimum moving distance cost of using the sensors emitted by base station $b_k$ to cover at least $K$ targets in interval $[x(t_l), x(t_i)]$. Denote the current targets of cover as $\tilde{\mathcal {T}}$, and the target covered by sensor $s$ as $\mathcal T_{s}$, we have the inner dynamic programming is

\begin{equation}\label{partial}
	d_{k}(l,i,K)=	\min\{\min\limits_{s\in G_{k}(l,i)\cap S_{i}}\{d_{k}(s)+d_{k}(l,r_s,K-|\mathcal T_{s} \setminus \tilde{\mathcal {T}}|)\},d_{k}(l,i-1,K)\}.
\end{equation}

The boundary conditions are increased by two on the basis of \ref{in_boundary},
\begin{equation}\label{eq1-1}
	\left\{
	\begin{array}{lcl}
		d_{k}(l,i,K)=+\infty, & &  i <l , K>0,\\
		d_{k}(l,i,K)=0, & &  i <l, K=0,\\
		d_{k}(l,i,K)=0, & & K\leq 0. \\
	\end{array} \right.
\end{equation}
The optimal solution to the partial MinMD$+$OCTC$_{line}$ problem can be obtained in $O(n^{5}m)$ time

\subsection{Polynomial-Time Algorithm for General MinMD$+$OCTC$_{line}$}\label{sec_ploy_alg_for_general}

For the optimal solution to the MinMD$+$OCTC$_{line}$ problem, the centers of the sensors all location on the line $L$. However, coverage can be achieved even if the center of the sensor is not necessarily located on a straight line. We note this as General MinMD$+$OCTC$_{line}$ problem. The key to design of dynamic programming algorithm for MinMD$+$OCTC$_{line}$ problem is that the number of the potential sensor positions in the optimal solution is polynomial, and Lemma \ref{Lem1}. We will show next that these two points still apply to general MinMD$+$OCTC$_{line}$ problems.

\begin{observation}\label{general_Ploy_sensor_position}
	In the optimal solution of General MinMD$+$OCTC$_{line}$ problem, the number of potential sensor locations are polynomial.
\end{observation}
\begin{proof}
 Similar to the proof of observation \ref{Ploy_sensor_position}, there is at least one target on the coverage boundary of the final location of all sensors. Therefore, we only need to explain that for any target, taking it as a target on the coverage boundary, the number of potential locations of sensors emitted by any base station is polynomial. Assume $t_i$ is a target on the boundary of the sensor emitted by $b_k$, $s$ is the sensor that is emitted by base station $b_k$ and has the minimum moving distance cost with $t_i$ on the coverage boundary. If the sensor covering $t_i$ is not sensor $s$ but sensor $s'$, there must be a target $t_j$ on the coverage boundary of sensor $s'$. Otherwise, the cost of moving distance can be reduced by moving closer to $s$. There are $O(n)$ situations. Therefore, for a base station $b_k$ in the optimal solution, the final number of locations of the sensors it emits is $O(n^2)$, total number is $O(mn^2)$.
	
\end{proof}

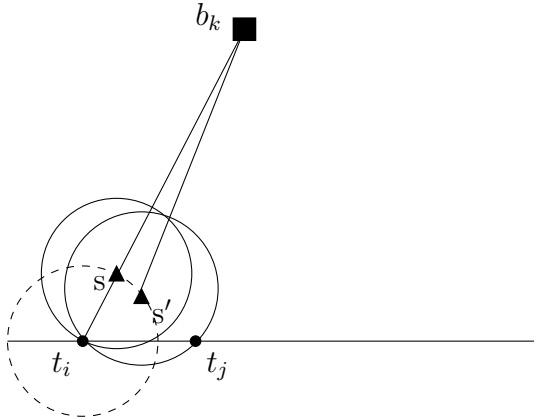
\begin{figure}[htbp]
	\centering
	\begin{tikzpicture}
		
		\draw (0,0) -- (7,0);

		\filldraw[black] (3,4) rectangle ++(0.3,0.3);

		\filldraw[black] (1,0) circle (2pt);
		\node at (1,0) [below left] {$t_i$};
        \node at (3,4) [above left] {$b_k$};		
		
		\draw[dashed] (1,0) circle (1);

		\draw (3.15,4.15) -- (1,0);
		\coordinate (intersection) at (1.45, 1);
		\node at (intersection) [below left] {$\text{s}$};
		
		\filldraw[black] (intersection) -- ++(-0.1,-0.2) -- ++(0.2,0) -- cycle;
		\draw (1.45,0.9) circle (1);
		
		\coordinate (intersection1) at (1.78, 0.7);
		\node at (intersection1) [below right] {$\text{s}'$};

		\filldraw[black] (intersection1) -- ++(-0.1,-0.2) -- ++(0.2,0) -- cycle;
		\draw (1.78,0.7) circle (1.02);
		\draw (3.15,4.15) -- (intersection1);
		\filldraw[black] (2.5,0) circle (2pt);
		\node at (2.5,0) [below right] {$t_j$};
	\end{tikzpicture}
	\caption{An illustration of the number of sensor locations being polynomial. }\label{fig0706-3}
\end{figure}

The first point is established for General MinMD$+$OCTC$_{line}$ problem.
 We can still only consider that any base station $b_k$ has $y(b_k)\geq 0$. The only difference in the proof is that for any sensor $s$ with coordinates $(x(s),y(s))$ in $OPT(\mathcal{I})$, if it is emitted by the base station $b_k$ with $y(b_k)>0$, place a sensor $s'$ emitted by $b_k'$ at the same position in the $\mathcal{I}'$ instance. Otherwise, place a sensor $s'$ emitted by $b_k'$ at $(x(s),|y(s)|)$ (For a base station $b_k$ with $y(b_k)<0$, the sensor $s$ it emits must have $y(s)\leq0$).

In Lemma \ref{Lem1}, we know that if the sensor is emitted by the base station $b_k$, then there is no other base station within the circle with the center of the sensor as the center and $d_k(s)$ as the radius. For General MinMD$+$OCTC$_{line}$ problem, we will illustrate that if the sensor $s$ emitted by base station $b_k$ in the optimal solution, and cover target interval $[x(t_i),x(t_j)]$, then there is an area where there are no opening other base stations.   
	
If a sensor with radius $r$ needs to cover both $t_i$ and $t_j$, then the sensor center must be located within $C_{t_i,r} \cap C_{t_j,r}$, which can be simplified to be located within the red curve $\gamma_{ij}$, since the sensors are all above the straight line. 
Because the sensor covering $[x(t_i),x(t_j)]$ in the optimal solution is emitted by base station $b_k$, this shows that among all the opened base stations in the optimal solution, the sensor emitted by $b_k$ is the sensor with the shortest moving distance that can cover the interval $[x(t_i),x(t_j)]$. Then within the range of $d_k(s)$ from the curve $\gamma_{ij}$, there cannot be any other base station. Otherwise, this base station can emit sensor $s'$ coverage $[x(t_i),x(t_j)]$ to replace the role of sensor $s$, and obtain a solution with a smaller moving distance cost. This is contradictory to the optimal solution. 

Now, we will explicitly draw the curve at a distance $d_k(s)$ from the curve $\gamma_{ij}$.  The curve $\gamma_{ij}$ is divided into four parts, arc $EF$, point $F$, arc $FG$ and straight line $GE$. The distance $d_k(s)$ from arc $EF$ is a part of circle $C_{t_j,d_k(s)+r}$, that is, arc $AB$. Similarly, for arc $FG$, we can find the arc $CD$ of circle $C_{t_i,d_k(s)+r}$, and the distance $d_k(s)$ from point $F$ is circle $C_{F,d_k(s)}$. The straight line $AD$ is because we only consider the situation above the straight line. Finally, we get the black curve $\Gamma_{k,ij}$ in the figure \ref{fig_general_inside_non_base}.

 \begin{figure}[htbp]
	\centering
	\begin{tikzpicture}
		
		\coordinate (A) at (-4,0);
		\coordinate (B) at (8,0);
		\coordinate (C) at (3.5,1.6);
		\coordinate (D) at (5,1.5);
		\coordinate (E) at (1,0);
		\coordinate (F) at (2.4,0);
		\coordinate (G) at (1.7,0.7);
		\coordinate (H) at (1.91,0.45);
		\coordinate (I1) at (-0.55,0);
		\coordinate (I2) at (0.34,2.1);
		\coordinate (I3) at (3.1,2.1);
		\coordinate (I4) at (3.95,0);

		\draw (A) -- (B);
    	\draw (C) -- (H);
		
	   \filldraw[black] (E) circle (2pt);
	   \filldraw[black] (F) circle (2pt);
	   \filldraw[black] (H) circle (2pt);
	   
		\filldraw[black] (C) ++(-0.15,-0.15) rectangle ++(0.3,0.3); 
		
		\draw[dashed] (E) circle (1);
		\draw[dashed] (F) circle (1);
		\draw (H) circle (1);
		\draw[dashed] (E) circle (2.96);
		\draw[dashed] (F) circle (2.96);
		\draw[dashed] (G) circle (1.96);
		
		\node at (C) [above right] {$b_k$};
		\node at (H) [right] {$\text{s}$};
		\node at (E) [above] {$t_i$};
		\node at (F) [above right] {$t_j$};
		\node at (I1) [above left] {$A$};
		\node at (I2) [above] {$B$};
		\node at (I3) [above] {$C$};
		\node at (I4) [above right] {$D$};
		\node at (1.4,0) [below] {$E$};
		\node at (2,0) [below] {$G$};
		\node at (G) [above] {$F$};
		\node at (-3,-1) {$C_{t_i,d_k(s)+r}$};
		\node at (6.5,-1) {$C_{t_j,d_k(s)+r}$};
		\node at (1.8,-1.5) {$C_{F,d_k(s)}$};
		
		\draw[line width=2pt] (E) ++(0:2.96) arc (0:46:2.96);
		\draw[line width=2pt] (G) ++(45.5:1.96) arc (45.5:134.5:1.96);
	    \draw[line width=2pt] (F) ++(134.5:2.96) arc (135:180:2.96);
	    \draw[line width=2pt, red] (E) ++(0:1) arc (0:46:1);
	    \draw[line width=2pt, red] (F) ++(134.5:1) arc (134.5:180:1);
	    \draw[line width=2pt] (I1) -- (I4) ;
	     \draw[line width=2pt,red ] (1.4,0) -- (2,0) ;
	    
	     \filldraw[red] (G) circle (2pt);
	     \filldraw[red] (1.4,0) circle (2pt);
	     \filldraw[red] (2,0) circle (2pt);
	    \filldraw[black] (I1) circle (2pt);
	    \filldraw[black] (I2) circle (2pt);
	    \filldraw[black] (I3) circle (2pt);
	    \filldraw[black] (I4) circle (2pt);
	\end{tikzpicture}
	\caption{ }\label{fig_general_inside_non_base}
\end{figure}

In summary, we can get the following observation.
\begin{observation}\label{general_sensor_location}
	For an optimal solution, target interval $[x(t_i),x(t_j)]$ is covered by sensor $s$ that is emitted by base station $b_k$,  then there will be no other base stations opened inside the curve $\Gamma_{k,ij}$ in the figure \ref{fig_general_inside_non_base}.
\end{observation}

 \begin{lemma}\label{Lem2} 
	Assume that $[x(t_i),x(t_j)]$ (resp. $[x(t_{i'}),x(t_{j'})]$) is the leftmost (resp. rightmost) target interval covered by the sensor $s_u$ (resp. $s_v$) emitted by base station $b_k$. Then for any target $x(t) \in [x(t_i),x(t_{j'})]$ , it is covered by the sensor from base station $b_k$.
 \end{lemma}

\begin{proof}

	\begin{figure}[htbp]
		\centering
		\begin{tikzpicture}
			\coordinate (A) at (-4,0);
			\coordinate (B) at (12,0);
			\coordinate (C) at (1.5,1.6);
			\coordinate (D) at (3,1.5);
			\coordinate (E) at (-1,0);
			\coordinate (F) at (0.4,0);
			\coordinate (G) at (-0.3,0.7);
			\coordinate (H) at (-0.15,0.55);
			\coordinate (I1) at (-2.55,0);
			\coordinate (I2) at (-1.7,2.1);
			\coordinate (I3) at (1.1,2.1);
			\coordinate (I4) at (1.95,0);

			\draw (A) -- (B);
			\draw (C) -- (H);
			
			\filldraw[black] (E) circle (2pt);
			\filldraw[black] (F) circle (2pt);
		
			\filldraw[black] (H) circle (2pt);
			
			\filldraw[black] (C) ++(-0.15,-0.15) rectangle ++(0.3,0.3);

			\draw (H) circle (1);

			\node at (1.4,1.7) [above] {$b_k$};
			\node at (H) [above] {$\text{s}_u$};
			\node at (E) [above] {$t_i$};
			\node at (F) [above] {$t_j$};

			\node at (I2) [above left] {$\Gamma_{k,ij}$};
			
			\draw[line width=2pt] (E) ++(0:2.95) arc (0:46:2.95);
			\draw[line width=2pt] (G) ++(45.5:1.96) arc (45.5:134.5:1.96);
			\draw[line width=2pt] (F) ++(134.5:2.95) arc (135:180:2.95);
			\draw[line width=2pt] (I1) -- (I4) ;
	
			\coordinate (E1) at (4.84,0);
			\coordinate (F1) at (6.8,0);
			\coordinate (G1) at (5.83,0.35);
			\coordinate (H1) at (5.8,0.2);
			\coordinate (I1') at (1.25,0);
			\coordinate (I2') at (1.87,2.5);
			\coordinate (I3') at (9.79,2.5);
			\coordinate (I4') at (10.4,0);
			
			\filldraw[black] (E1) circle (2pt);
			\filldraw[black] (F1) circle (2pt);

			\filldraw[black] (H1) circle (2pt);

			\draw (H1) circle (1);

			\draw (C) -- (H1);
			
			\node at (H1) [above] {$\text{s}_v$};
			\node at (E1) [above left] {$t_{i'}$};
			\node at (F1) [above right] {$t_{j'}$};
			
			\draw[line width=2pt] (E1) ++(0:5.54) arc (0:27:5.54);
			\draw[line width=2pt] (G1) ++(28:4.5) arc (28:152:4.5);
			\draw[line width=2pt] (F1) ++(153:5.525) arc (153:180:5.525);
			\draw[line width=2pt] (I1') -- (I4');

			\node at (9,3.5) [above right] {$\Gamma_{k,s_i'j'}$};

			\coordinate (E11) at (2.4,0);
			\coordinate (F11) at (3.6,0);
			\coordinate (G11) at (3,0.8);
			\coordinate (H11) at (2.83,0.6);
			\coordinate (I1'') at (1.25,0);
			\coordinate (I2'') at (1.87,2.5);
			\coordinate (I3'') at (9.79,2.5);
			\coordinate (I4'') at (10.4,0);
			
			\filldraw[black] (E11) circle (2pt);
			\filldraw[black] (F11) circle (2pt);

			\filldraw[black] (H11) circle (2pt);

			\draw (H11) circle (1);

			\draw (C) -- (H11);
			
			\node at (H11) [above] {$\text{s}_w$};
			\node at (E11) [above] {$t_{i''}$};
			\node at (F11) [above] {$t_{j''}$};
			\node at (4.5,2)  [above] {$\Gamma_{k,i''j''}$};
			
			\draw[line width=2pt] (E11) ++(0:2.664) arc (0:53:2.664);
			\draw[line width=2pt] (G11) ++(53:1.664) arc (53:127:1.664);
			\draw[line width=2pt] (F11) ++(127:2.664) arc (127:180:2.664);
			\draw[line width=2pt] (I1'') -- (I4'');
		
		\end{tikzpicture}
		\caption{ }\label{fig_general_consecutive_interval}
	\end{figure}
	
Since $s_u$ and $s_v$ are emitted by base station $b_k$, according to Observation \ref{general_sensor_location}, we have that there are no base stations within curves $\Gamma_{k,ij}$ and $\Gamma_{k,i'j'}$. 
 Assuming that in the optimal solution, there exist a target interval $[x(t_{i''}),x(t_{j''})] \subset [x(t_i),x(t_{j'})]$ which covered by another sensor $s$ form base station $b_k'$. In addition, if $[x(t_{i''}),x(t_{j''})]$ is covered by sensor emitted by $b_k$, we can find the sensor $s_w$ by finding the closest point form base station $b_k$ to curve $\gamma_{i''j''}$ and get the curve $\Gamma_{k,i''j''}$.  Then  $d_{k'}(s)\leq d_k(s_w)$. If we define the interior of the curve $\Gamma$ as ``range $\Gamma$'', then we have the base station $b_{k'}$ is located in the curve $\Gamma_{k,i''j''}$ or range $\Gamma_{k,i''j''}$.  However, it can be seen from the figure \ref{fig_general_consecutive_interval} that

 $$\text{range}~\Gamma_{k,i''j''} \subseteq \text{range}~\Gamma_{k,ij} \cup \text{range}~\Gamma_{k,i'j'}.$$
 
We already know that there is no any base station except $b_k$ in $\text{range}~\Gamma_{k,ij} \cup \text{range}~\Gamma_{k,i'j'}$. Therefore, base station $b_{k'}$ can only be located at the location of base station $b_{k}$. However, any sensor emitted by $b_{k'}$ can be emitted by base station $b_{k}$, in order to reduce costs, we will not open base station $b_{k'}$.
 So, any target $x(t) \in [x(t_i),x(t_{j'})]$ is covered by the sensor from base station $b_k$.
\end{proof}

Therefore, there is an optimal solution of General MinMD$+$OCTC$_{line}$ problem, the target interval that any opening the base station $b_k$ is responsible for is consecutive. The property of midpoint perpendicular is still established for any two base station. Then, the second point is finished. We have the conclusion that the dynamic programming can be used to solve General MinMD$+$OCTC$_{line}$ problem.  Due to the increase in the number of potential sensor locations, the current running time is $O(mn^5)$.

\subsection{Approximation algorithm for a special case of MinMD$+$OCTC}\label{2.3}
In MinMD$+$OCTC problem, the target is distributed in the plane, and the opened base station can emit the sensor to cover it, which goal is to minimize the cost of opened base stations and the moving distance of sensors. In this section,
for a special case of the MinMD$+$OCTC problem, that is, when the base station is  far from the target, we give an approximation of 8.928.

The MinMD$+$OCTC problem algorithm design we gave is based on the algorithm of the UFL problem. For an instance $\mathcal{I}=(\mathcal{T},\mathcal{B})$, We regard the target $t_i$ as the customer $C_i$ in the UFL problem and the base station $b_k$ as the facility $F_k$.  After the UFL algorithm, if customer $C_i$ is served by facility $F_k$ in the UFL problem, then we can get a feasible solution for MinMD$+$OCTC problem by  placing a sensor emitted by base station $b_k$ in $\mathcal{I}$ at target $t_i$. However, when the target distribution is dense, the quality of this solution will be very poor. So our idea is to use some token targets to replace the target of the dense area, and then use the algorithm of the UFL problem. 
The plane is divided into a hexagon grids with a side length of $r$, and the targets in the grid are rounded to the center point of this grid, the grid center is token target, we mark it as instance $\widetilde{\mathcal{I}}=(\widetilde{\mathcal{T}},\mathcal{B})$.  

Before the algorithm analysis, we show a conclusion that a circle on any plane can be completely covered by 6 grid circles, which was proved in \cite{gao2022efficient}. Here we use geometric properties to improve it to 5.

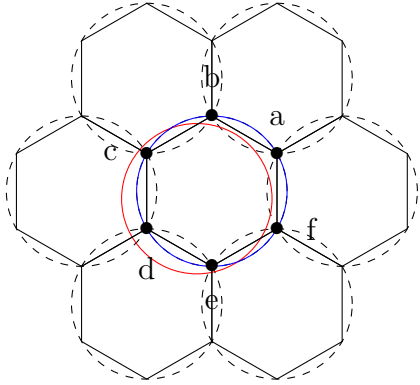
\begin{figure}[htbp]\label{hexagon}
	\centering
	\begin{tikzpicture}

		\def\radius{1}

		\pgfmathsetmacro{\sqrtthree}{sqrt(3)}

		\draw[dashed] (0,0) circle (\radius);

		\draw[dashed] (-\sqrtthree,0) circle (\radius);
		\draw[dashed] (\sqrtthree,0) circle (\radius);
		\draw[dashed] (\sqrtthree/2,3/2) circle (\radius);
		\draw[dashed] (-\sqrtthree/2,3/2) circle (\radius);
		\draw[dashed] (\sqrtthree/2,-3/2) circle (\radius);
		\draw[dashed] (-\sqrtthree/2,-3/2) circle (\radius);
		\draw[red] (-0.2,-0.1) circle (\radius);
		\draw[blue] (0,0)  circle (1);

		\node[label=above:a] at (30:\radius) {\textbullet};
		\node[label=above:b] at (90:\radius) {\textbullet};
		\node[label=left:c] at (150:\radius) {\textbullet};
		\node[label=below:d] at (210:\radius) {\textbullet};
		\node[label=below:e] at (270:\radius) {\textbullet};
		\node[label=right:f] at (330:\radius) {\textbullet};

		\draw (30:\radius) -- (90:\radius) -- (150:\radius) -- (210:\radius) -- (270:\radius) -- (330:\radius) -- cycle;

		\draw (-\sqrtthree,0) + (30:\radius) -- + (90:\radius) -- + (150:\radius) -- + (210:\radius) -- + (270:\radius) -- + (330:\radius) -- cycle;
		\draw (\sqrtthree,0) + (30:\radius) -- + (90:\radius) -- + (150:\radius) -- + (210:\radius) -- + (270:\radius) -- + (330:\radius) -- cycle;
		\draw (\sqrtthree/2,3/2) + (30:\radius) -- + (90:\radius) -- + (150:\radius) -- + (210:\radius) -- + (270:\radius) -- + (330:\radius) -- cycle;
		\draw (-\sqrtthree/2,3/2) + (30:\radius) -- + (90:\radius) -- + (150:\radius) -- + (210:\radius) -- + (270:\radius) -- + (330:\radius) -- cycle;
		\draw (\sqrtthree/2,-3/2) + (30:\radius) -- + (90:\radius) -- + (150:\radius) -- + (210:\radius) -- + (270:\radius) -- + (330:\radius) -- cycle;
		\draw (-\sqrtthree/2,-3/2) + (30:\radius) -- + (90:\radius) -- + (150:\radius) -- + (210:\radius) -- + (270:\radius) -- + (330:\radius) -- cycle;
		
	\end{tikzpicture}
	
	\caption{An illustration for the Lemma \ref{Lem_5_grid}. }\label{fig2}
\end{figure}

\begin{lemma}\label{Lem_5_grid}
	For any circle with radius $r$ on the plane, it is completely covered by at most 5 grid circles with radius $r$.
\end{lemma}

\begin{proof}
 For any circle on the plane, it is related to at most 7 grid circles, as shown in Figure \ref{fig2}. A key observation is that if a circle needs to be covered with a grid circle if and only if the circle covers at least one corner point of the grid circle. If a circle contains a bunch of corner points of a hexagonal grid (such as point $a$ and $d$), then a grid circle can completely cover this circle, such as the blue circle in Figure \ref{fig2}. Otherwise it contains at most three consecutive corner points, this means that at most 5 of the 7 grid circles are used. 
\end{proof}

\begin{theorem}\label{opt}
	There is an $8.928$ approximation algorithm for MinMD$+$OCTC problem.
\end{theorem}
\begin{proof}

	We record the optimal and feasible solutions of instance $\mathcal{I}$ as $OPT$ and $APX$, and for instance $\widetilde{\mathcal{I}}$ as $OPT_g$ and $APX_g$, respectively.  We can get a feasible solution $APX_g(\widetilde{\mathcal{I}})$ by applying the UFL algorithm (placing a sensor at each token target), and record the sensor as {\em grid sensor}. $APX_g$ also covers all target, take it as feasible solution for instance $\mathcal{I}$. In order to analyze the gap between $APX_g$ and $OPT$, we transform $OPT$ into a feasible solution $APX_g^{trans}$ of $\widetilde{\mathcal{I}}$. First, base station $b_k$ is opened in $APX_g^{trans}$ if and only if $b_k$ is opened in $OPT$. Secondly, according to Lemma \ref{Lem_5_grid}, a sensor $s$ emitted by base station $b_k$ in $OPT$ will be replaced with 5 grid sensors $s_1,...,s_5$, let $s_1,...,s_5$ are emitted by $b_k$. Assuming that the base station is far away from the target, that is, $\min \|t_i b_k\|>10r$, then we have the following conclusions. 
		$$\sum\limits_{i=1}^{5}d_{APX_g^{trans}}(s_{i})<5d_{OPT}(s)+10r <6d_{OPT}(s).$$
	Therefore, 
	\begin{align}
		apx_g^{trans} &= \sum\limits_{k \in APX_g^{trans}}c_k+ \sum\limits_{i \in APX_g^{trans}}d(s_i) \nonumber \\
		 &= \sum\limits_{k \in OPT}c_k+ \sum\limits_{i \in APX_g^{trans}}d(s_i) \nonumber \\
		 &\leq  \sum\limits_{k \in OPT}c_k+ 6\sum\limits_{i \in OPT}d(s_i) \nonumber \\
		\nonumber
		&\leq  6opt \\
		\nonumber
	\end{align}
	
Since $APX_g^{trans}$ is a feasible solution of instance $\widetilde{\mathcal{I}}$, and 1.488 approximation algorithm for instance $\widetilde{\mathcal{I}}$, one has
$$apx=apx_g\leq 1.488 opt_g \leq 1.488 apx_g^{trans}.$$

To sum up, we get $apx \leq 8.928 opt$.
	
\end{proof}

\section{Conclusion}
In this paper, we consider the Minimum Moving-Distance and Opening-Costs Target Coverage problem (MinMD$+$OCTC) problem. For the case where the target is on a straight line, we use dynamic programming to give the optimal solution. In particular, we also provide a polynomial time algorithm for the case where the sensor center does not need to move to a straight line. When the target is on the plane, and give a simple algorithm with an approximate ratio of 8.928. The key to designing algorithms for the MinMD$+$OCTC problem in this paper is to use the geometric properties of the Euclidean space. The result of the article is based on the smaller radius of the sensor relative to the moving distance. So one challenge is how to remove this assumption. In addition, there is a problem is how to design algorithms for the MinMD$+$OCTC problem when the radius of the sensor is different. However, we find that when the radius is different, Lemma \ref{Lem1} is not satisfied. This is also a difficult point and worth solving.

\bibliographystyle{plain}

\bibliography{Ref1}

\end{document}